\documentclass[10pt]{article}

\usepackage{graphics}
\usepackage{graphicx}

\usepackage[a4paper, left=35mm,right=35mm,top=34mm,bottom=34mm]{geometry}
\usepackage[utf8]{inputenc}
\usepackage[T1]{fontenc}
\usepackage[english]{babel}

\usepackage{enumerate}
\usepackage{graphicx}
\usepackage{hyperref}
\hypersetup{
    colorlinks=true,
    linkcolor=blue,
    filecolor=magenta,      
    urlcolor=cyan,
}

\usepackage{mathtools,amsthm,amssymb,amsfonts}
\usepackage{algorithm}
\usepackage{algorithmic}
\makeatother
\theoremstyle{plain}
\newtheorem{theorem}{Theorem}
\newtheorem{corollary}{Corollary}
\newtheorem{proposition}{Proposition}
\theoremstyle{definition}

\theoremstyle{remark}
\newtheorem*{remark}{Remark}

\usepackage{caption} 
\usepackage{fancyhdr}

\author{
  {\normalsize Fr\'ed\'eric Magoul\`es}\thanks{CentraleSup\'elec, Universit\'e Paris-Saclay, France
    (correspondence, frederic.magoules@hotmail.com).}
  \and
  {\normalsize Guillaume Gbikpi-Benissan}\thanks{IRT SystemX, France.}
}
\title{Asynchronous parareal time discretization for partial differential equations}
\date{}

\begin{document}
\maketitle
\thispagestyle{fancy}

\begin{abstract}
\noindent Asynchronous iterations are more and more investigated for both scaling and fault-resilience purpose on high performance computing platforms. While so far, they have been exclusively applied within space domain decomposition frameworks, this paper advocates a novel application direction targeting time-decomposed time-parallel approaches. Specifically, an asynchronous iterative model is derived from the Parareal scheme, for which convergence and speedup analysis are then conducted. It turned out that Parareal and async-Parareal feature very close convergence conditions, asymptotically equivalent, including the finite-time termination property. Based on a computational cost model aware of unsteady communication delays, our speedup analysis shows the potential performance gain from asynchronous iterations, which is confirmed by some experimental case of heat evolution on a homogeneous supercomputer. This primary work clearly suggests possible further benefits from asynchronous iterations.
\end{abstract}


\section{Introduction}

While high-performance computing (HPC) architectures are more and more parallel, the development of efficient and robust algorithms that can fully exploit high degrees of parallelism is still an actual challenge in many scientific computing fields. For solving partial differential equations (PDEs) defined on the interior of arbitrary geometrical shapes, strong advances resulted in highly scalable domain decomposition methods (DDMs) (see, e.g., \cite{Lions89, LeTallec:1994:DMC}), which feature high rates of speedup and efficiency. Further, although it first appears unnatural to break the time line causality in time-dependent differential equations, thorough applications of DDMs still led researchers to time-parallel methods for allowing actual parallel time-integration iterations. Waveform relaxation methods \cite{LelEtAl1982}, for instance, solve space-time PDEs by means of independent time-integration processes performed on subregions of the spatial domain, while ensuring global consistency through time-dependent interface conditions.
Multiple shooting \cite{CharPhil1993} and time-multigrid \cite{Leimkuhler1998} methods led to the more recent Parareal time discretization scheme \cite{LionsEtAl2001}, and a general analysis of parallel implicit time-integration algorithms \cite{FarChan2003}, where the time domain itself is decomposed. Basically, two levels of time discretization grid are considered. An approximated solution of the PDE on a coarse grid provides initial values on each time subinterval, when more accurate solutions are thus independently computed on the basis of a sufficiently fine time grid. These fine solutions are then used to enhance new another coarse solution. In this sense, the Parareal computational model can also be considered as an iterative predictor-corrector scheme. We refer to \cite{GanVan2007} for an introductory insight into time-decomposed time-integration approaches.

Despite such great advances with DDMs, there remain scaling limits upon these parallel algorithms, due to the fact that they require, at some point, sequential global iterations, even if each of these iterations is performed in a parallel fashion. This sequence requirement implies iterative global synchronization, which causes both bottleneck and fault-resilience issues at HPC scales. For such methods then, the only way to completely break sequential computation is to resort to asynchronous iterations \cite{BertTsit1991}, which therefore constitute today one of the most interesting algorithmic approaches toward HPC applications. Arisen from linear relaxation schemes (see \cite{ChazMir1969}), asynchronous iterations are basically described by a non-deterministic computational model, such that there is no more a single sequence of parallel iterations for approximating a global solution. Instead, one can independently run a proper sequence of iterations on each subregion of a domain, therefore communication can be completely overlapped by meaningful computation. So far, only fixed-point iterative models have been proved to be asynchronously convergent, under various general contraction conditions (see, e.g., \cite{Miellou1975, Baudet1978, ElTarazi1982, Bertsekas1983}). Still, despite the lower convergence rate of asynchronous iterations, their practical efficiency in execution time over classical fixed-point methods is undoubtedly established, thanks to extensive studies addressing various computational problems (see, e.g., \cite{LiShEtAl1985, HartMcCor1989, ChajZen1991, ChauEtAl2007, ChauEtAl2011}).
The asynchronous iterations theory has been largely applied within the framework of overlapping Schwarz domain decomposition methods (see, e.g., \cite{FromEtAl1997, SpitEtAl2001, SpitEtAl2003}), according to their tight relation with block-Jacobi and Gauss-Seidel relaxation schemes (see, e.g., \cite{Gand2008}). This included waveform relaxation methods as well (see, e.g., \cite{BahiEtAl1996, FromSzyld2000}), which constitutes so far the only application targeting time-parallel methods. Very recently, non-overlapping Schwarz methods have been investigated (see \cite{MagEtAl2017}), just as some spatial domain substructuring approach (see \cite{MagVen2018}). The goal of this paper is to introduce a novel asynchronous iterations application direction by addressing time-decomposed time-parallel algorithms through the well-known Parareal iterative scheme. Besides a successful experimental investigation in \cite{MagEtAl2018}, where we only focused on implementation aspects, we present here a theoretical analysis to fully assess both convergence conditions and speedup potential.

Section \ref{sec:async} provides a short general introduction to asynchronous iterations, with most important analysis tools for assessing their application to the Parareal iterative scheme. In section \ref{sec:apr} then, we deal with its asynchronous convergence analysis, after deriving the corresponding iterative model. Speedup analysis is conducted in section \ref{sec:perf}, where we first propose a computational cost model not neglecting communication effects. Section \ref{sec:res} finally presents practical performance experimented on a supercomputer, for a heat evolution problem, and our conclusions follow in section \ref{sec:conclu}.

\section{General asynchronous iterations}
\label{sec:async}

\subsection{Computational model}

Asynchronous iterations arose in the framework of linear relaxation methods, from the successive experimental and theoretical works of Rosenfeld \cite{Rosenfeld1969} and Chazan \& Miranker \cite{ChazMir1969}, respectively. Given a linear system
\[
A x = b,
\]
where $x$ is a vector of $n$ unknowns, let $M$ be a nonsingular matrix and consider the mapping $f$ defined by
\[
f(x) = (I - M^{-1} A) x + M^{-1} b,
\]
with $I$ denoting the identity matrix. Given then that
\[
A = M - (M - A),
\]
this linear system can be seen as a fixed-point problem, as we have
\[
A x = b \quad \iff \quad x = f(x).
\]
For instance, the Jacobi relaxation scheme consists of taking $M$ as the diagonal part of $A$. A classical iterative fixed point search generates a sequence $\{x^{k}\}_{k \in \mathbb{N}}$ of vectors such that
\begin{equation}
\label{eq:si}
x^{k+1} = f(x^{k}),
\end{equation}
and such that, for any given initial vector $x^{0}$,
\[
\lim_{k \to +\infty} x^{k} = x^{*},
\]
with $x^{*}$ being the unique fixed point of $f$. In a parallel computation context with $p$ processors, $p \le n$, one may assume a decomposition of the form
{
\[
x = (x_{1}, \ldots, x_{p})^{\mathsf T}, \qquad f(x) = (f_{1}(x), \ldots, f_{p}(x))^{\mathsf T},
\]
}
which gives an equivalent sequence of parallel iterations
\[
x_{i}^{k+1} = f_{i}(x_{1}^{k}, \ldots, x_{p}^{k}), \qquad \forall i \in \{1, \ldots, p\}.
\]
The first step toward asynchronous iterations is the introduction of the free steering (see, e.g., \cite{Schechter1959}), where, at each iteration, one is allowed to relax an arbitrary subset of unknowns. Note however that parallel computing were not yet targeted at this time, and that, except for free steering, the choice of the unknown(s) to be relaxed was based on deterministic criteria (like Gauss-Seidel relaxation, for instance, which follows a cyclic pattern). Let us then consider, at each iteration, an arbitrary subset
\[
P^{(k)} \subseteq \{1, \ldots, p\},
\]
defining the relaxed subvectors of unknowns. The corresponding free-steering iterative model is given by
\[
x_{i}^{k+1} = \left \{
\begin{array}{ll}
f_{i}\left(x_{1}^{k}, \ldots, x_{p}^{k}\right), & \forall i \in P^{(k)},\\
x_{i}^{k}, & \forall i \notin P^{(k)},
\end{array}
\right.
\]
where iterations
\[
x_{i}^{k+1} = x_{i}^{k}, \quad \forall i \notin P^{(k)},
\]
are only implicit.
Then, while considering actual parallel free steering, one would be interested in letting each subvector be continuously relaxed at a random pace. In this case therefore, there is no more explicit global iterations, and the sequence $\{P^{(k)}\}_{k \in \mathbb{N}}$ instead models implicit global iterations by accounting each concurrent update of at least one subvector. In such a parallel context, Chazan \& Miranker \cite{ChazMir1969} took into account the delay of accessing a subvector, in the sense that this subvector could be relaxed in the meantime, and therefore could become outdated before being used. Consequently, a subvector $x_{i}^{k+1}$, $i \in \{1, \ldots, p\}$, is no more necessarily computed on the basis of subvectors $x_{j}^{k}$, $j \in \{1, \ldots, p\}$, but on the basis of possibly outdated subvectors $x_{j}^{\tau_{j}^{i}(k)}$, i.e. with
\[
\tau_{j}^{i}(k) \le k.
\]
This yielded asynchronous iterations
\begin{equation}
\label{eq:ai}
x_{i}^{k+1} = \left \{
\begin{array}{ll}
f_{i}\left(x_{1}^{\tau_{1}^{i}(k)}, \ldots, x_{p}^{\tau_{p}^{i}(k)}\right), & \forall i \in P^{(k)},\\
x_{i}^{k}, & \forall i \notin P^{(k)}.
\end{array}
\right.
\end{equation}
Two straight assumptions are made to allow for convergence. First, no component $i \in \{1, \ldots, p\}$ should definitively stop being updated, which implies that $i$ belongs to successive subsets $P^{(k)}$ infinitely often, i.e.
\begin{equation}
\label{eq:ai_ass1}
\forall i \in \{1, \dots, p\}, \quad \operatorname{card}\{k \in \mathbb{N} \ | \ i \in P^{(k)}\} = +\infty,
\end{equation}
where $\operatorname{card}$ denotes the cardinality. Secondly, delays to access updated components are bounded, so that
\begin{equation}
\label{eq:ai_ass2}
\forall i, j \in \{1, \dots, p\}, \quad \lim_{k \to +\infty} \tau_{j}^{i}(k) = +\infty.
\end{equation}
Note finally that classical iterations are thus specialization of asynchronous iterations where
\[
P^{(k)} = \{1, \ldots, p\}, \quad \tau_{j}^{i}(k) = k, \qquad \forall k \in \mathbb{N}.
\]

More generally now, we can consider an arbitrary, non necessarily linear, mapping
\[
\widetilde f : E^{m} \to E, \qquad m \in \mathbb{N},
\]
with $E^{m}$ being the Cartesian product of $m$ sets $E = E_{1} \times \cdots \times E_{p}$, and a fixed-point problem
\[
x = \widetilde f(x, x, \ldots, x).
\]
The study of such a generalization goes back to Baudet \cite{Baudet1978} for addressing Newton asynchronous iterations. This generalizes the model \eqref{eq:ai} to
\begin{equation}
\label{eq:gai}
x_{i}^{k+1} = \left \{
\begin{array}{ll}
\widetilde f_{i}\left(x_{1}^{\tau_{1,1}^{i}(k)}, \ldots, x_{p}^{\tau_{p,1}^{i}(k)}, \ldots, x_{1}^{\tau_{1,m}^{i}(k)}, \ldots, x_{p}^{\tau_{p,m}^{i}(k)}\right), & \forall i \in P^{(k)},\\
x_{i}^{k}, & \forall i \notin P^{(k)},
\end{array}
\right.
\end{equation}
in the sense that \eqref{eq:ai} corresponds to the case $m = 1$.

\subsection{Convergence conditions}

It is well known (see, e.g., Theorem 2.1 in \cite{BahiEtAl2007}) that the sequential iterative model \eqref{eq:si} is convergent if and only if
\[
\rho(I - M^{-1} A) < 1,
\]
where $\rho(.)$ denotes the spectral radius. Chazan \& Miranker \cite{ChazMir1969} showed that the asynchronous iterative model \eqref{eq:ai} is convergent if and only if
\[
\rho(|I - M^{-1} A|) < 1,
\]
where $|.|$ denotes the absolute value, which is slightly more restrictive, since we always have (see, e.g., Corollary 6.3 in \cite{BertTsit1989})
\[
\rho(I - M^{-1} A) \le \rho(|I - M^{-1} A|).
\]
Still, to study the application of asynchronous iterations to the Parareal scheme, we would rather need the framework of arbitrary mappings, for which we recall here two most general sufficient contraction conditions guaranteeing asynchronous convergence.

Let each set $E_{i}$, with $i \in \{1, \ldots, p\}$, be a Banach space provided with some norm $\|.\|_{(i)}$, and consider weighted maximum norms
\begin{equation}
\label{eq:wmn}
\|x\|_{\infty}^{w} = \max_{i \in \{1, \ldots, p\}} \frac{\|x_{i}\|_{(i)}}{w_{i}},
\end{equation}
where $w$ is a positive vector with real entries $w_{i}$, $i \in \{1, \ldots, p\}$. The following result goes back to El Tarazi \cite{ElTarazi1982}:
\begin{theorem}[El Tarazi, 1982]
\label{theo:et1982b}
The asynchronous iterative model \eqref{eq:gai} is convergent if there exists a positive vector $w$ and a positive real $\alpha < 1$ such that
\[
\forall X, Y \in E^{m}, \quad \|\widetilde f(X) - \widetilde f(Y)\|_{\infty}^{w} \le \alpha \max \{\|x^{(1)} - y^{(1)}\|_{\infty}^{w}, \ldots, \|x^{(m)} - y^{(m)}\|_{\infty}^{w}\},
\]
with
\[
X = (x^{(1)}, \ldots, x^{(m)}), \quad Y = (y^{(1)}, \ldots, y^{(m)}).
\]
\end{theorem}

Bertsekas \cite{Bertsekas1983} then provided a more general contraction framework based on the identification of an implicit sequence of nested sets which converges toward the solution set $\{x^{*}\}$ (actually, this even applies for mappings admitting several fixed points):
\begin{theorem}[Bertsekas, 1983]
\label{theo:b1983}
The asynchronous iterative model \eqref{eq:gai}, with
\[
m = 1,
\]
is convergent if there exists a sequence $\{S^{(l)}\}_{l \in \mathbb{N}}$ of subsets of $E$, with
\[
x^{0} \in S^{(0)}, \qquad S^{(l)} = S^{(l)}_{1} \times \cdots \times S^{(l)}_{p}, \qquad S^{(l+1)} \subset S^{(l)}, \qquad \lim_{l \to +\infty} S^{(l)} = \{x^{*}\},
\]
and such that
\[
\forall x \in S^{(l)}, \quad \widetilde f(x) \in S^{(l+1)}.
\]
\end{theorem}
As mentioned above, the sets sequence $\{S^{(l)}\}_{l \in \mathbb{N}}$ is only implicitly generated all along the successive updates of the components $x_{i}$, and, for more clarity, we do use another sequence variable, $l$, since several asynchronous iterations $k$ could be performed within a same set $S^{(l)}$ before entering the next one, $S^{(l+1)}$. The main difficulty in applying such a result to a specific mapping $\widetilde f$ is to be able to identify the sets $S^{(l)}$ and $S^{(l)}_{i}$. Bertsekas \& Tsitsiklis \cite{BertTsit1991} showed that this was more general than the contraction framework based on weighted maximum norms. When, for instance, Theorem \ref{theo:et1982b} is applicable, with $m = 1$, then we have (see \cite[section 6.3]{BertTsit1989} and \cite{BertTsit1991})
\begin{equation}
\label{eq:nsets}
S^{(l)} = \left\{x \in E \ | \ \|x - x^{*}\|_{\infty}^{w} \le \alpha^{l} \|x^{0} - x^{*}\|_{\infty}^{w}\right\}, \qquad l \in \mathbb{N},
\end{equation}
and
\[
S^{(l)}_{i} = \left\{\dot{x} \in E_{i} \ | \ \frac{\|\dot{x} - x_{i}^{*}\|_{(i)}}{w_{i}} \le \alpha^{l} \|x^{0} - x^{*}\|_{\infty}^{w}\right\}, \qquad l \in \mathbb{N}, \quad i \in \left\{1, \ldots, p\right\}.
\]

\section{Parareal asynchronous iterations}
\label{sec:apr}

\subsection{Computational model}

Let us consider a time-dependent PDE,
\begin{equation}
\label{eq:pde}
\frac{\partial u}{\partial t} (s, t) = g(u, s, t), \qquad t \in [0, T] \subset \mathbb{R}, \quad s \in \Omega \subset \mathbb{R}^{3},
\end{equation}
with some given boundary condition,
\[
h(u, s, t) = 0, \qquad t \in [0, T], \quad s \in \partial\Omega,
\]
and initial values,
\[
u^{(0)}(s) := u(s, 0), \qquad s \in \Omega.
\]
Let the time domain $[0, T]$ be given by
\[
[0, T] = [T_{0}, T_{1}] \cup [T_{1}, T_{2}] \cup \cdots \cup [T_{p-1}, T_{p}], \qquad p \in \mathbb{N},
\]
and assume initial values $\lambda_{i-1}(s)$, with $s \in \Omega$, on each subdomain $[T_{i-1}, T_{i}]$, with $i \in \{1, \ldots, p\}$. Independent subproblems can then be formulated by
\begin{equation}
\label{eq:sub_pde}
\left \{
\begin{array}{lllll}
\dfrac{\partial u_{i}}{\partial t} (s, t) & = & g(u_{i}, s, t), & t \in [T_{i-1}, T_{i}], & s \in \Omega,\\
h(u_{i}, s, t) & = & 0, & t \in [T_{i-1}, T_{i}], & s \in \partial\Omega,\\
u_{i}(s, T_{i-1}) & = & \lambda_{i-1}(s), & s \in \Omega, &
\end{array}
\right.
\end{equation}
which is equivalent to the global problem \eqref{eq:pde} if and only if
\begin{equation}
\label{eq:itf_cond}
\left \{
\begin{array}{llll}
\lambda_{0}(s) & = & u^{(0)}(s), &\\
\lambda_{i-1}(s) & = & u_{i-1}(s, T_{i-1}), & \forall i \in \{2, \ldots, p\}.
\end{array}
\right.
\end{equation}

Now let, for $i \in \{1, \ldots, p\}$,
\[
u_{i}^{(t)}(s) := u_{i}(s, t), \qquad s \in \Omega, \quad t \in [T_{i-1}, T_{i}],
\]
and let $F$ be a function which, to initial values $\lambda_{i-1}(s)$, associates a sufficiently accurate solution of the subproblem \eqref{eq:sub_pde}, at time $T_{i}$, i.e.,
\[
F(\lambda_{i-1}) \cong u_{i}^{(T_{i})}.
\]
Similarly, let $G$ be a function which, to initial values $\lambda_{i-1}(s)$, associates a cheap approximation of the solution of the subproblem \eqref{eq:sub_pde}, at time $T_{i}$, i.e.,
\[
G(\lambda_{i-1}) \approx u_{i}^{(T_{i})}.
\]
From initial guess
\[
\left \{
\begin{array}{lcll}
\lambda_{0}^{0} & = & u^{(0)}, &\\
\lambda_{i}^{0} & = & G(\lambda_{i-1}^{0}), & \forall i \in \{1, \ldots, p\},
\end{array}
\right.
\]
the Parareal iterative scheme \cite{LionsEtAl2001, FarChan2003} defines sequences $\{\lambda_{i}^{k}\}_{k \in \mathbb{N}}$ such that
\begin{equation}
\label{eq:pr}
\left \{
\begin{array}{lcll}
\lambda_{0}^{k+1} & = & \lambda_{0}^{k}, &\\
\lambda_{i}^{k+1} & = & G(\lambda_{i-1}^{k+1}) + F(\lambda_{i-1}^{k}) - G(\lambda_{i-1}^{k}), & \forall i \in \{1, \ldots, p\},
\end{array}
\right.
\end{equation}
which should converge to a collection $\{\lambda_{i}^{*}\}$ satisfying the interface conditions \eqref{eq:itf_cond}. This algorithm is thus expected to provide the solution $u^{(T)}$ of \eqref{eq:pde} at time $T$, without applying the whole sequence
{
\[
u^{(T)} = F \circ F \circ \cdots \circ F(u^{(0)}).
\]
}

To easily assess such a computational model within the previous framework of general asynchronous iterations, let us consider linear operators
\[
G(\lambda_{i}) = \mathcal G \lambda_{i} + \delta, \quad F(\lambda_{i}) = \mathcal F \lambda_{i} + \zeta, \qquad i \in \{0, 1, \ldots, p\},
\]
and a vector
\[
\lambda = (\lambda_{0}, \lambda_{1}, \ldots, \lambda_{p}).
\]
Then, \eqref{eq:pr} simplifies to
\[
\left \{
\begin{array}{lcll}
\lambda_{0}^{k+1} & = & \lambda_{0}^{k}, &\\
\lambda_{i}^{k+1} & = & \mathcal G \lambda_{i-1}^{k+1} + (\mathcal F - \mathcal G) \lambda_{i-1}^{k} + \zeta, & \forall i \in \{1, \ldots, p\},
\end{array}
\right.
\]
which corresponds to synchronous relaxations \eqref{eq:si} of a problem
\[
A \lambda = b,
\]
with
\[
A =
\begin{bmatrix}
I & 0 & 0 & \cdots & 0\\
-\mathcal F & I & 0 & \cdots & 0\\
0 & -\mathcal F & I & \ddots & \vdots\\
\vdots & \ddots & \ddots & \ddots & 0\\
0 & \cdots & 0 & -\mathcal F & I
\end{bmatrix},
\quad
M =
\begin{bmatrix}
I & 0 & 0 & \cdots & 0\\
-\mathcal G & I & 0 & \cdots & 0\\
0 & -\mathcal G & I & \ddots & \vdots\\
\vdots & \ddots & \ddots & \ddots & 0\\
0 & \cdots & 0 & -\mathcal G & I
\end{bmatrix},
\quad
b =
\begin{bmatrix}
\lambda_{0}^{*}\\
\zeta\\
\zeta\\
\vdots\\
\zeta
\end{bmatrix},
\]
recalling that the iterations mapping $f$ is given by
\[
f(\lambda) = (I - M^{-1}A) \lambda + M^{-1} b.
\]
\begin{remark}
This corresponds to a preconditioned Richardson iterative scheme, as also formulated, e.g., in \cite{NielHest2016}.
\end{remark}
Nonetheless, a straightforward application of asynchronous relaxations \eqref{eq:ai} would yield a computation of the form
\[
\lambda_{i}^{k+1} = \left\{
\begin{array}{ll}
G(\lambda_{i-1}^{k+1}) + F(\lambda_{i-1}^{\tau_{i-1}^{i}(k)}) - G(\lambda_{i-1}^{\tau_{i-1}^{i}(k)}), & \forall i \in P^{(k)},\\
\lambda_{i}^{k}, & \forall i \notin P^{(k)},
\end{array}
\right.
\]
which is not of great interest, due to the synchronization on $G(\lambda_{i-1}^{k+1})$. We would therefore rather be interested in seeing what happens with a computation, at least, of the form
\begin{equation}
\label{eq:apr}
\lambda_{i}^{k+1} = \left\{
\begin{array}{ll}
G(\lambda_{i-1}^{\tau_{i-1,1}^{i}(k)}) + F(\lambda_{i-1}^{\tau_{i-1,2}^{i}(k)}) - G(\lambda_{i-1}^{\tau_{i-1,2}^{i}(k)}), & \forall i \in P^{(k)},\\
\lambda_{i}^{k}, & \forall i \notin P^{(k)},
\end{array}
\right.
\end{equation}
which corresponds to the more general asynchronous iterations \eqref{eq:gai}, with $m = 2$, and where the iterations mapping $\widetilde f$ is given by
\begin{equation}
\label{eq:apr_f}
\left\{
\begin{array}{lcll}
\widetilde f_{0}(\lambda^{(1)}, \lambda^{(2)}) & = & \lambda_{0}^{*}, &\\
\widetilde f_{i}(\lambda^{(1)}, \lambda^{(2)}) & = & G(\lambda_{i-1}^{(1)}) + F(\lambda_{i-1}^{(2)}) - G(\lambda_{i-1}^{(2)}), & \forall i \in \{1, \ldots, p\}.
\end{array}
\right.
\end{equation}
It is easy to verify that
\[
\lambda = f(\lambda) \quad \iff \quad \lambda = \widetilde f(\lambda, \lambda),
\]
which ensures an equivalent fixed-point problem. We recall that the asynchronous iterations variable $k$ grows faster than the synchronous iterations one, therefore, the delays expressed in \eqref{eq:apr} do not have to be seen as implying a retarded iterative procedure.

Note, finally, that the mapping $\widetilde f$ is explicitly expressed without linearity assumption on $F$ and $G$. In the following convergence analysis of the Parareal asynchronous iterations \eqref{eq:apr}, such an assumption will be sometimes made only to keep the development simple and to use same notations as with the linear synchronous formulation. Furthermore, to fully exploit the domain decomposition point of view, one will notice that, in the asynchronous framework, the extension of these results to a general formulation using $F_{i}$ and $G_{i}$ is quite straightforward.

\subsection{Convergence conditions}

Let $\|.\|_{\infty}$ denote a maximum norm \eqref{eq:wmn} where $w$ is a vector of $1$. We first restate preliminary results in our linear synchronous framework.
\begin{theorem}
\label{theo:pr_error}
The Parareal synchronous iterative model \eqref{eq:pr} defines a sequence $\{\lambda^{k}\}_{k \in \mathbb{N}}$ such that
\[
\|\lambda^{k} - \lambda^{*}\|_{\infty} \le \alpha^{k} \|\lambda^{0} - \lambda^{*}\|_{\infty},
\]
with
\[
\alpha = \frac{1 - \theta^{p}}{1 - \theta} \|\mathcal F - \mathcal G\|, \qquad \theta \ne 1, \quad \theta \ge \|\mathcal G\|,
\]
and $\|.\|$ denoting any given norm.
\end{theorem}
\begin{proof}
See Theorem 4.5 in \cite{GanVan2007} for a possible proof, or directly consider the fact that linear fixed-point formulation implies
\[
\|\lambda^{k} - \lambda^{*}\|_{\infty} \le (\|I - M^{-1} A\|_{\infty})^{k} \|\lambda^{0} - \lambda^{*}\|_{\infty}.
\]
\end{proof}
\begin{remark}
\[
\lim_{p \to +\infty} \frac{1 - \theta^{p}}{1 - \theta} \|\mathcal F-\mathcal G\| = \left\{
\begin{array}{ll}
+\infty, & \theta > 1,\\
\dfrac{1}{1 - \theta} \|\mathcal F-\mathcal G\|, & \theta < 1.
\end{array}
\right.
\]
One is therefore interested by the case $\theta < 1$ which ensures a stable iterative procedure.
\end{remark}
\begin{corollary}
\label{cor:pr_conv}
The Parareal synchronous iterative model \eqref{eq:pr} is convergent if
\[
\|\mathcal G\| < 1,
\]
and
\[
\|\mathcal G\| + \|\mathcal F-\mathcal G\| < 1 + \|\mathcal G\|^{p} \|\mathcal F-\mathcal G\|.
\]
\end{corollary}
\begin{proof}
It directly follows from Theorem \ref{theo:pr_error} that convergence is guaranteed if
\[
\frac{1 - \theta^{p}}{1 - \theta} \|\mathcal F - \mathcal G\| < 1.
\]
If $\theta < 1$, this is implied by
\[
(1 - \theta^{p}) \|\mathcal F - \mathcal G\| < 1 - \theta.
\]
If then $\|\mathcal G\| < 1$, we can take $\theta = \|\mathcal G\|$, which leads to the result.
\end{proof}
\begin{remark}
The Parareal synchronous iterative model \eqref{eq:pr} is thus asymptotically convergent ($p \to +\infty$) if
\[
\|\mathcal G\| + \|\mathcal F-\mathcal G\| < 1.
\]
\end{remark}

We need here some words about the terminology. It is well established (see \cite{FarChan2003}) that Parareal iterations terminate in finite time $k = p$, due to the sequential propagation of the accurate solution, from $[T_{0}, T_{1}]$ to $[T_{p-1}, T_{p}]$. By ``convergence'', we therefore mean the fact that the iterative procedure actually allows one to reach a given precision by decreasing the maximum error over the whole time domain $[T_{0}, T_{p}]$, which thus suggests the possibility to converge before getting sequential. To be rigorous though, ``convergence'', here, is a shortcut word for ``convergence with respect to the maximum norm''.

In the following, we would now like to extend such results to the asynchronous framework.
\begin{theorem}
\label{theo:apr_error}
The Parareal asynchronous iterative model \eqref{eq:apr} defines a sequence $\{\lambda^{k}\}_{k \in \mathbb{N}}$ such that
\[
\|\lambda^{k} - \lambda^{*}\|_{\infty} \le \widetilde\alpha^{\sigma(k)} \|\lambda^{0} - \lambda^{*}\|_{\infty},
\]
with
\[
\widetilde\alpha = \|\mathcal G\| + \|\mathcal F-\mathcal G\|,
\]
where $\sigma$ is an integer-valued function on $\mathbb{N}$ satisfying
\[
\underset{k \to +\infty}{\lim} \sigma(k) = +\infty.
\]
\end{theorem}
\begin{proof}
According to \eqref{eq:apr}, we have, $\forall i \in P^{(k)}$,
\[
\begin{aligned}
\lambda_{i}^{k+1} & = \mathcal G \lambda_{i-1}^{\tau_{i-1,1}^{i}(k)} + (\mathcal F - \mathcal G) \lambda_{i-1}^{\tau_{i-1,2}^{i}(k)} + \zeta,\\
\lambda_{i}^{k+1} - \lambda_{i}^{*} & = \mathcal G (\lambda_{i-1}^{\tau_{i-1,1}^{i}(k)} - \lambda_{i-1}^{*}) + (\mathcal F - \mathcal G) (\lambda_{i-1}^{\tau_{i-1,2}^{i}(k)} - \lambda_{i-1}^{*}),\\
\|\lambda_{i}^{k+1} - \lambda_{i}^{*}\| & \le \|\mathcal G\| \|\lambda_{i-1}^{\tau_{i-1,1}^{i}(k)} - \lambda_{i-1}^{*}\| + \|\mathcal F - \mathcal G\| \|\lambda_{i-1}^{\tau_{i-1,2}^{i}(k)} - \lambda_{i-1}^{*}\|,\\
\|\lambda_{i}^{k+1} - \lambda_{i}^{*}\| & \le (\|\mathcal G\| + \|\mathcal F - \mathcal G\|) \max \{\|\lambda_{i-1}^{\tau_{i-1,1}^{i}(k)} - \lambda_{i-1}^{*}\|, \|\lambda_{i-1}^{\tau_{i-1,2}^{i}(k)} - \lambda_{i-1}^{*}\|\},
\end{aligned}
\]
which gives
\[
\left \{
\begin{array}{lcll}
\|\lambda_{i}^{k+1} - \lambda_{i}^{*}\| & \le & \widetilde \alpha \max \{\|\lambda_{i-1}^{\tau_{i-1,1}^{i}(k)} - \lambda_{i-1}^{*}\|, \|\lambda_{i-1}^{\tau_{i-1,2}^{i}(k)} - \lambda_{i-1}^{*}\|\}, & \forall i \in P^{(k)},\\
\|\lambda_{i}^{k+1} - \lambda_{i}^{*}\| & = & \|\lambda_{i}^{k} - \lambda_{i}^{*}\|, & \forall i \notin P^{(k)}.
\end{array}
\right.
\]
It follows that, $\forall i \in \{1, \ldots, p\}$,
\[
\|\lambda_{i}^{k+1} - \lambda_{i}^{*}\| \le \max \left\{
\begin{array}{l}
\widetilde \alpha \max \{\|\lambda_{i-1}^{\tau_{i-1,1}^{i}(k)} - \lambda_{i-1}^{*}\|, \|\lambda_{i-1}^{\tau_{i-1,2}^{i}(k)} - \lambda_{i-1}^{*}\|\},\\
\|\lambda_{i}^{k} - \lambda_{i}^{*}\|
\end{array}
\right\},
\]
and then,
\begin{equation}
\label{eq:apr_error}
\|\lambda^{k+1} - \lambda^{*}\|_{\infty} \le \max_{i \in \{1, \ldots, p\}} \left\{
\begin{array}{l}
\widetilde \alpha \|\lambda_{i-1}^{\tau_{i-1,1}^{i}(k)} - \lambda_{i-1}^{*}\|,\\
\widetilde \alpha \|\lambda_{i-1}^{\tau_{i-1,2}^{i}(k)} - \lambda_{i-1}^{*}\|,\\
\|\lambda_{i}^{k} - \lambda_{i}^{*}\|
\end{array}
\right\}.
\end{equation}

Let us define sets
\[
S^{(k)} = \{\lambda \ | \ \|\lambda - \lambda^{*}\|_{\infty} \le \widetilde\alpha^{\sigma(k)} \|\lambda^{0} - \lambda^{*}\|_{\infty}\},
\]
with $\sigma$ being a function on $\mathbb{N}$ such that
\[
\sigma(0) = 0.
\]
We obviously verify
\[
\lambda^{0} \in S^{(0)}.
\]
Assume then that
\[
\forall l \le k, \quad \lambda^{l} \in S^{(l)}.
\]
Therefore, from \eqref{eq:apr_error}, it follows that
\[
\|\lambda^{k+1} - \lambda^{*}\|_{\infty} \le \max_{i \in \{1, \ldots, p\}} \left\{
\begin{array}{l}
\widetilde \alpha \widetilde\alpha^{\sigma(\tau_{i-1,1}^{i}(k))} \|\lambda^{0} - \lambda^{*}\|_{\infty},\\
\widetilde \alpha \widetilde\alpha^{\sigma(\tau_{i-1,2}^{i}(k))} \|\lambda^{0} - \lambda^{*}\|_{\infty},\\
\widetilde\alpha^{\sigma(k)} \|\lambda^{0} - \lambda^{*}\|_{\infty}
\end{array}
\right\},
\]
and by taking $\sigma(k+1)$ such that
\begin{equation}
\label{eq:apr_error_2}
\widetilde\alpha^{\sigma(k+1)} = \max_{i \in \{1, \ldots, p\}} \left\{\widetilde\alpha^{\sigma(\tau_{i-1,1}^{i}(k))+1}, \widetilde\alpha^{\sigma(\tau_{i-1,2}^{i}(k))+1}, \widetilde\alpha^{\sigma(k)}\right\},
\end{equation}
we do have
\[
\|\lambda^{k+1} - \lambda^{*}\|_{\infty} \le \widetilde \alpha^{\sigma(k+1)} \|\lambda^{0} - \lambda^{*}\|_{\infty},
\]
which confirms that
\[
\forall k \in \mathbb{N}, \quad \lambda^{k} \in S^{(k)}.
\]

Finally, according to the classical assumption \eqref{eq:ai_ass1} on $P^{k}$, we eventually have, $\forall i \notin P^{(k)}$,
\[
\lambda_{i}^{k+1} = \lambda_{i}^{k} = \lambda_{i}^{k-1} = \cdots = \lambda_{i}^{k_{i}+1} = \mathcal G \lambda_{i-1}^{\tau_{i-1,1}^{i}(k_{i})} + (\mathcal F - \mathcal G) \lambda_{i-1}^{\tau_{i-1,2}^{i}(k_{i})} + \zeta, \quad i \in P^{(k_{i})},
\]
which implies in \eqref{eq:apr_error_2} that
\[
\sigma(k+1) \in \left\{
\begin{array}{l}
\sigma(\tau_{i-1,1}^{i}(k))+1, \ \ \sigma(\tau_{i-1,2}^{i}(k))+1,\\
\sigma(\tau_{i-1,1}^{i}(k_{i}))+1, \ \ \sigma(\tau_{i-1,2}^{i}(k_{i}))+1
\end{array}
\right\}_{i \in \{1, \ldots, p\}}.
\]
Therefore, with the other classical assumption \eqref{eq:ai_ass2} on $\tau_{i-1,1}^{i}$ and $\tau_{i-1,2}^{i}$, we also eventually have
\[
\sigma(k+1) = \sigma(k-d_{1})+1 = \sigma(k-d_{2})+1+1 = \cdots,
\]
with
\[
0 \le d_{1} < d_{2} < \cdots \le k,
\]
which ensures that
\[
\underset{k \to +\infty}{\lim} \sigma(k) = +\infty,
\]
and concludes the proof.
\end{proof}
With such an error bound, a convergence condition trivially follows.
\begin{corollary}
\label{cor:apr_conv}
The Parareal asynchronous iterative model \eqref{eq:apr} is convergent if
\begin{equation}
\label{eq:apr_conv}
\|\mathcal G\| + \|\mathcal F-\mathcal G\| < 1.
\end{equation}
\end{corollary}
\begin{proof}
This is straightforward, from Theorem \ref{theo:apr_error}.
\end{proof}
\begin{remark}
The proof of Theorem \ref{theo:apr_error} exhibits
\[
\|\lambda_{i}^{k+1} - \lambda_{i}^{*}\| \le (\|\mathcal G\| + \|\mathcal F - \mathcal G\|) \max \{\|\lambda_{i-1}^{\tau_{i-1,1}^{i}(k)} - \lambda_{i-1}^{*}\|, \|\lambda_{i-1}^{\tau_{i-1,2}^{i}(k)} - \lambda_{i-1}^{*}\|\},
\]
which clearly indicates that Parareal asynchronous iterations \eqref{eq:apr} match the contraction framework of Theorem \ref{theo:et1982b}, which also implies convergence under the same sufficient condition \eqref{eq:apr_conv}.
\end{remark}

One of the advantages of asynchronous iterations is that the algorithm can be studied on the basis of only local behavior $\widetilde f_{i}$, independently of its input which, therefore, can be provided by whatever other local contracting mapping $\widetilde f_{i-1}$. This is an interesting feature which suggests much more flexibility for further developments. Indeed, for instance, in the proof of Theorem \ref{theo:apr_error}, one can notice that the whole analysis is based on only
\[
\lambda_{i}^{k+1} = \mathcal G \lambda_{i-1}^{\tau_{i-1,1}^{i}(k)} + (\mathcal F - \mathcal G) \lambda_{i-1}^{\tau_{i-1,2}^{i}(k)} + \zeta.
\]
Therefore, to fully account for the domain decomposition context, the resulting error bound can be directly restated using $\mathcal G_{i}$ and $\mathcal F_{i}$, with
\[
\widetilde\alpha = \max_{i} \|\mathcal G_{i}\| + \|\mathcal F_{i}-\mathcal G_{i}\|.
\]
Furthermore, without linearity assumption on functions $G_{i}$ and $F_{i}$, the same result can be simply obtained, for
\[
\lambda_{i}^{k+1} = G_{i}(\lambda_{i-1}^{\tau_{i-1,1}^{i}(k)}) + F_{i}(\lambda_{i-1}^{\tau_{i-1,2}^{i}(k)}) - G_{i}(\lambda_{i-1}^{\tau_{i-1,2}^{i}(k)}),
\]
by generalizing the arbitrary norm $\|.\|$ to a distance function over the space of functions $u^{(t)}$, with $u^{(t)}(s) := u(s,t)$.

In conclusion, it is interesting to see that, while asynchronous iterations generally require more restrictive convergence conditions, they can be successfully performed in any Parareal-based solver, for large numbers of time subintervals (${p \to \infty}$), without verification of additional properties, since convergence conditions are very nearly the same.

When, additionally, it happens that the Parareal iterative scheme itself does not actually converge, a finite-time termination is still ensured, at $k = p$, with the solution provided by a sequential application of $F$ from $[T_{0}, T_{1}]$ to $[T_{p-1}, T_{p}]$ (see \cite{FarChan2003}). We use here the idea of nested sets $\{S^{(l)}\}_{l \in \mathbb{N}}$ from Theorem \ref{theo:b1983} to show the finite-time termination of Parareal asynchronous iterations. Note that deriving nested sets of the form \eqref{eq:nsets} applies only when the convergence condition is satisfied, what is not assumed here for finite-time termination. We therefore instead define $S^{(l)}$ so as to have the following property.
\begin{proposition}
\label{prop:pras_nset}
Parareal asynchronous iterations \eqref{eq:apr} generate a sets sequence $\{S^{(l)}\}_{l \in \mathbb{N}}$ such that
\begin{equation}
\label{eq:pras_nset}
\forall l \ge p, \quad S^{(l)} = \{\lambda^{*}\}.
\end{equation}
\end{proposition}
\begin{proof}
Let us define sets
\[
S^{(l)} = \left\{\lambda \ | \ \forall i \in \{0, \ldots, \min\{l, p\}\}, \quad \lambda_{i} = \lambda_{i}^{*}\right\}, \quad\quad l \in \mathbb{N},
\]
so that we have \eqref{eq:pras_nset} and
\begin{equation}
\label{eq:apr_ns1}
\lambda^{0} \in S^{(0)}.
\end{equation}
Now let
\[
\lambda^{(1)}, \lambda^{(2)} \in S^{(l)}, \qquad \lambda^{(3)} = \widetilde f(\lambda^{(1)}, \lambda^{(2)}),
\]
where $\widetilde f$ is the Parareal asynchronous iterations mapping (see \eqref{eq:apr_f}). Then we have
\[
\left\{
\begin{array}{lcll}
\lambda^{(3)}_{0} & = & \lambda_{0}^{*}, &\\
\lambda^{(3)}_{i} & = & G(\lambda_{i-1}^{*}) + F(\lambda_{i-1}^{*}) - G(\lambda_{i-1}^{*}) = \lambda_{i}^{*}, & \forall i \in \{1, \ldots, \min\{l+1, p\}\},
\end{array}
\right.
\]
which implies that
\begin{equation}
\label{eq:apr_ns2a}
\forall \lambda^{(1)}, \lambda^{(2)} \in S^{(l)}, \quad \widetilde f(\lambda^{(1)}, \lambda^{(2)}) \in S^{(l+1)}.
\end{equation}

At last, let $\Gamma$ denotes the set of all functions $u^{(t)}$, with $u^{(t)}(s) := u(s,t)$, and let us define sets
\[
S^{(l)}_{i} = \left\{
\begin{array}{lll}
\{\lambda_{i}^{*}\}, & i \in \{0, \ldots, \min\{l, p\}\}, & l \in \mathbb{N},\\
\Gamma, & i \in \{l+1, \ldots, p\}, & l \in \{0, \ldots, p-1\},
\end{array}
\right.
\]
so that we do verify
\begin{equation}
\label{eq:apr_ns3}
S^{(l)} = S^{(l)}_{0} \times \cdots \times S^{(l)}_{p}.
\end{equation}

\eqref{eq:apr_ns1} to \eqref{eq:apr_ns3} ensure that Parareal asynchronous iterations \eqref{eq:apr} generate such a sequence $\{S^{(l)}\}_{l \in \mathbb{N}}$, which concludes the proof.
\end{proof}

\section{Performance analysis}
\label{sec:perf}

\subsection{Computational cost}

Farhat \& Chandesris \cite{FarChan2003} established a bound of $1/2$ on the parallel efficiency (speedup per processor) of the Parareal scheme, assuming that at least one corrective iteration would be needed (which implies at least $k=2$ iterations). It would therefore be of some interest to see at which extent asynchronous iterations could improve the parallel performance, since it is intuitively expected that, here as well, the finite-time termination causes the possible time saving to be bounded.

Without questioning this efficiency bound, Aubanel \cite{Aub2011} provided additional insights by taking into account different possible implementations. We consider here the distributed approach therein proposed, which features better scaling potential. Basically, it consists of distributing the sequential application of $G$ over the set of processors, as described by Algorithm \ref{algo:pr_siac}. We assume the case where we have $p$ processors, so that $\lambda_{i}$, with $i \in \{1, \ldots, p\}$, is held by the processor $i$.
\begin{algorithm}[htbp]
\caption{Distributed Parareal scheme (on processor $i \in \{1, \ldots, p\}$)}
\label{algo:pr_siac}
{\small
\begin{algorithmic}[1]
\STATE{$\lambda_{i-1}^{0} := u^{(0)}$}
\FORALL{$j \in \{1, \ldots, i-1\}$}
	\STATE{$\lambda_{i-1}^{0} := G(\lambda_{i-1}^{0})$}
\ENDFOR
\STATE{$w_{i}^{0} := G(\lambda_{i-1}^{0})$}
\STATE{$\lambda_{i}^{0} := w_{i}^{0}$}
\STATE{$k := 0$}
\STATE{Allow for message reception from processor $i-1$}
\REPEAT
\IF{$i > k$}
	\STATE{$v_{i}^{k} := F(\lambda_{i-1}^{k})$}
	\IF{$i-1 > k$}
		\STATE{Wait for reception of $\lambda_{i-1}^{k+1}$ from processor $i-1$}
		\STATE{$w_{i}^{k+1} := G(\lambda_{i-1}^{k+1})$}
	\ELSE
		\STATE{$w_{i}^{k+1} := w_{i}^{k}$}
	\ENDIF
	\STATE{$\lambda_{i}^{k+1} := w_{i}^{k+1} + v_{i}^{k} - w_{i}^{k}$}
	\IF{$i < p$}
		\STATE{Request sending of $\lambda_{i}^{k+1}$ to processor $i+1$}
	\ENDIF
\ELSE
	\STATE{$\lambda_{i}^{k+1} := \lambda_{i}^{k}$}
\ENDIF
	\STATE{$k := k + 1$}
\UNTIL{$\lambda^{k} \simeq \lambda^{*}$}
\end{algorithmic}
}
\end{algorithm}
Let then $\mathcal{C}_{G}$ and $\mathcal{C}_{F}$ denote the computational costs of applying $G$ and $F$, respectively. The cost of a sequential application of $F$ over the whole time interval is thus given by
\[
\mathcal{C}(p, \mathcal{C}_{F}) = p\ \mathcal{C}_{F}.
\]
What is remarkable with such a distributed implementation is that the processor $i-1$ starts applying $F$ before the processor $i$. Therefore, in a perfect load balance context, one may expect $\lambda_{i-1}^{k+1}$ to be available to the processor $i$ at the moment when it finishes applying $F$, what would result in a computational cost, after $k$ iterations, given by (see \cite{Aub2011})
\[
\mathcal{C}^{(k)}(p, \mathcal{C}_{F}, \mathcal{C}_{G}) = p\ \mathcal{C}_{G} + k (\mathcal{C}_{F} + \mathcal{C}_{G}),
\]
where $p\ \mathcal{C}_{G}$ corresponds to the cost of the initialization phase in Algorithm \ref{algo:pr_siac}.

While accurate enough within the scope of \cite{Aub2011}, this ideal case however assumes negligible communication, vector addition and memory access costs, which is {questionable in real world} situations. Let then, at least, $\mathcal{C}_{C,G}(i,k)$ give communication overhead costs which accumulate with $\mathcal{C}_{G}$ on the processor $i$, at the iteration $k$, without being overlapped by the application of $F$ on the processor $p$. It yields:
\begin{equation}
\label{eq:pr_cost}
\mathcal{C}^{(k)}(p, \mathcal{C}_{F}, \mathcal{C}_{G}, \mathcal{C}_{C,G}) = p\ \mathcal{C}_{G} + k (\mathcal{C}_{F} + \mathcal{C}_{G}) + \sum_{l=1}^{k} \sum_{j=l+1}^{p-1} \mathcal{C}_{C,G}(j,l).
\end{equation}
\begin{remark}
Actually, not only communication cost is not negligible, but it can even be higher than $\mathcal C_{G}$, since both are closely proportional to the size of $\lambda_{i}$. Note that $\mathcal{C}_{C,G}(j,l)$ is of the form
\[
\mathcal{C}_{C,G}(j,l) = \mathcal{C}_{C}(j,l) + \beta_{j,l}\ \mathcal{C}_{G}, \qquad \beta_{j,l} \in \{0, 1\},
\]
with
\[
\mathcal{C}_{C}(j,l) > 0 \quad \iff \quad \beta_{j,l} = 1,
\]
where $\mathcal{C}_{C}(j,l)$ is a non-overlapped communication cost.
\end{remark}
Still, {let us assume a constant uniform communication cost and, for instance, a four-steps fine integration
\[
F = \mathring F \circ \mathring F \circ \mathring F \circ \mathring F,
\]
with $\mathring F$ being some one-step integration which costs nearly as much as $G$. Then, by introducing non-negligible communication $C$ into Aubanel's model \cite{Aub2011},}
Algorithm \ref{algo:pr_siac} would behave like
\[
\begin{array}{cclcccllll}
\text{proc 1} & G \mathring F \mathring F \mathring F \mathring F, &&&&&&&&\\
 & & C &&&&&&&\\
\text{proc 2} & G G \mathring F \mathring F \mathring F & \mathring F G & \mathring F & \mathring F & \mathring F & \mathring F, &&&\\
 & & & C & & & & C & &\\
\text{proc 3} & G G G \mathring F \mathring F & \mathring F \mathring F & - & G & \mathring F & \mathring F & \mathring F \mathring F G & \mathring F \mathring F \mathring F \mathring F, &\\
 & & & & & C & & & C & C\\
\text{proc 4} & G G G G \mathring F & \mathring F \mathring F & \mathring F & - & - & G & \mathring F \mathring F \mathring F & \mathring F G \mathring F \mathring F & \mathring F \mathring F G \mathring F \mathring F \mathring F \mathring F,
\end{array}
\]
which shows that the overhead would lie only in the initialization or the first iteration phase, implying that
\[
\sum_{j=l+1}^{p-1} \mathcal{C}_{C,G}(j,l) = 0, \qquad \forall l > 1.
\]
Our general cost model including $\mathcal{C}_{C,G}$ therefore becomes handy when variations can occur in the communication delays, due for instance to a fault recovery procedure or simple perturbation on the network. We see from the example that a little cost increase in the second message transfer toward the processor 4 would result in an overhead cost. At the same time, the impact of these variations can also be expected to decrease with $k$. Nevertheless, instead of an exhaustive analysis, we are rather interested in assessing a global potential effect of communication. Let us then consider a rough estimation, $\mathcal{\overline C}_{C,G}$, of the average overhead cost at each iteration, on each processor. This leads to:
\begin{proposition}
\label{prop:pr_cost}
The computational cost of $k$ distributed Parareal iterations (Algorithm \ref{algo:pr_siac}) is of the form
\[
\mathcal{C}^{(k)}(p, \mathcal{C}_{F}, \mathcal{C}_{G}, \mathcal{\overline C}_{C,G}) = p\ \mathcal{C}_{G} + k \left(\mathcal{C}_{F} + \mathcal{C}_{G} + \left(p-1 - \frac{k+1}{2}\right) \mathcal{\overline C}_{C,G}\right).
\]
\end{proposition}
\begin{proof}
From \eqref{eq:pr_cost}, we have
\[
\begin{aligned}
\mathcal{C}^{(k)}(p, \mathcal{C}_{F}, \mathcal{C}_{G}, \mathcal{\overline C}_{C,G}) & = p\ \mathcal{C}_{G} + k (\mathcal{C}_{F} + \mathcal{C}_{G}) + \sum_{l=1}^{k} (p-1 - l)\ \mathcal{\overline C}_{C,G}\\
& = p\ \mathcal{C}_{G} + k (\mathcal{C}_{F} + \mathcal{C}_{G}) + \left(k(p-1) - \frac{k(k+1)}{2}\right) \mathcal{\overline C}_{C,G}.
\end{aligned}
\]
The proposition thus follows.
\end{proof}
\begin{remark}
While $k$ might be expected to remain relatively small, the effects of $\mathcal {\overline C}_{C,G}$ however increase with the number $p$ of processors. It is important to note that accounting such overhead costs does not contradict the perfect load balance assumption. In this sense, this would be a minimum relevant ideal model for the distributed Parareal scheme.
\end{remark}

Let us finally introduce asynchronous iterations into Algorithm \ref{algo:pr_siac}, to obtain Algorithm \ref{algo:pr_aiac}.
\begin{algorithm}[htbp]
\caption{Distributed async-Parareal scheme (on processor $i \in \{1, \ldots, p\}$)}
\label{algo:pr_aiac}
{\small
\begin{algorithmic}[1]
\STATE{$\lambda_{i-1}^{0} := u^{(0)}$}
\FORALL{$j \in \{1, \ldots, i-1\}$}
	\STATE{$\lambda_{i-1}^{0} := G(\lambda_{i-1}^{0})$}
\ENDFOR
\STATE{$w_{i}^{0} := G(\lambda_{i-1}^{0})$}
\STATE{$\lambda_{i}^{0} := w_{i}^{0}$}
\STATE{$k_{i} := 0$}
\STATE{Allow for message reception from processor $i-1$}
\REPEAT
	\STATE{$v_{i}^{k_{i}} := F(\lambda_{i-1}^{k_{i}})$}
		\STATE{$w_{i}^{k_{i}+1} := G(\lambda_{i-1}^{k_{i}+1})$}
	\STATE{$\lambda_{i}^{k_{i}+1} := w_{i}^{k_{i}+1} + v_{i}^{k_{i}} - w_{i}^{k_{i}}$}
	\IF{$i < p$}
		\STATE{Request sending of $\lambda_{i}^{k_{i}+1}$ to processor $i+1$}
	\ENDIF
	\STATE{$k_{i} := k_{i} + 1$}
\UNTIL{$(\lambda_{0}^{0}, \lambda_{1}^{k_{1}}, \ldots, \lambda_{p}^{k_{p}}) \simeq \lambda^{*}$}
\end{algorithmic}
}
\end{algorithm}
Still assuming perfect load balance, the corresponding computational cost is that of the processor having performed the most iterations. To make connection with the implicit {global asynchronous iterations variable, say $\widetilde k$}, let us define, for $k_{i}$ iterations on processor $i$,
\[
\kappa_{i}(\widetilde k) := k_{i} = \operatorname{card}\{l \le \widetilde k \ | \ i \in P^{(l)}\},
\]
and set
\[
\kappa(\widetilde k) := \max_{i \in \{1, \ldots, p\}} \kappa_{i}(\widetilde k).
\]
Then we have:
\begin{proposition}
\label{prop:apr_cost}
The computational cost of $\widetilde k$ distributed Parareal asynchronous iterations (Algorithm \ref{algo:pr_aiac}) is of the form
\[
\mathcal{\widetilde C}^{(\widetilde k)}(p, \mathcal{C}_{F}, \mathcal{C}_{G}) = p\ \mathcal{C}_{G} + \kappa(\widetilde k) (\mathcal{C}_{F} + \mathcal{C}_{G}).
\]
\end{proposition}
\begin{proof}
This is straightforward, since communication is completely overlapped.
\end{proof}

\subsection{Performance gain}

The error bounds from Theorem \ref{theo:pr_error} and Theorem \ref{theo:apr_error} indicate that the convergence speeds of Parareal synchronous and asynchronous iterations depend on how small the respective factors
\[
\alpha = \frac{1 - \|\mathcal G\|^{p}}{1 - \|\mathcal G\|} \|\mathcal F - \mathcal G\|, \qquad \widetilde\alpha = \|\mathcal G\| + \|\mathcal F-\mathcal G\|
\]
are. As intuitively expected, we verify the following comparison.
\begin{proposition}
\label{prop:conv_rate}
\[
\widetilde \alpha < 1 \quad \implies \quad \alpha < \widetilde \alpha.
\]
\end{proposition}
\begin{proof}
Since $\|\mathcal G\| < 1$, we have
\[
\alpha < \frac{1}{1 - \|\mathcal G\|} \|\mathcal F - \mathcal G\|.
\]
Assume that
\[
\widetilde \alpha \le \frac{1}{1 - \|\mathcal G\|} \|\mathcal F - \mathcal G\|.
\]
Then we have
\[
\begin{aligned}
(\|\mathcal G\| + \|\mathcal F-\mathcal G\|) (1 - \|\mathcal G\|) & \le \|\mathcal F - \mathcal G\|,\\
\|\mathcal G\| (1 - \|\mathcal G\|) + \|\mathcal F-\mathcal G\| (1 - \|\mathcal G\| - 1) & \le 0,\\
\|\mathcal G\| (1 - \|\mathcal G\| - \|\mathcal F-\mathcal G\|) & \le 0,\\
1 - \|\mathcal G\| - \|\mathcal F-\mathcal G\| & \le 0,\\
\widetilde \alpha & \ge 1.
\end{aligned}
\]
Consequently,
\[
\widetilde \alpha < 1 \quad \implies \quad \frac{1}{1 - \|\mathcal G\|} \|\mathcal F - \mathcal G\| < \widetilde \alpha,
\]
which concludes the proof.
\end{proof}

Nonetheless, if we consider synchronous and asynchronous iterations variables $k$ and $\widetilde k$, respectively, we should bear in mind that $\widetilde k$ grows much faster than $k$. Therefore, we can reasonably expect that $\sigma(\widetilde k)$ also grows faster than $k$, sufficiently to make Parareal asynchronous iterations reach the desired precision sooner than synchronous ones do.

To be more practical, besides convergence rate, we see from Algorithm \ref{algo:pr_aiac} that, either $\lambda_{i-1}^{k_{i}+1}$ is always an updated available input for the processor $i$, and hence, this reduces to a classical Parareal execution, or it happens that
\[
\lambda_{i-1}^{k_{i}+1} = \lambda_{i-1}^{k_{i}}
\]
for some processors $i$, and hence, one has
\[
\lambda_{i}^{k_{i}+1} = F(\lambda_{i-1}^{k_{i}}).
\]
The advantage of the asynchronous scheme is to account for the fact that this could be enough to reach the desired precision, instead of waiting for an updated input in order to necessarily output an exact Parareal update. Note that, in any case then, we should have
\[
\kappa(\widetilde k) \ge k,
\]
where $k$ would be the number of synchronous iterations. From our cost model in Proposition \ref{prop:pr_cost} and Proposition \ref{prop:apr_cost}, we deduce:
\begin{corollary}
\label{cor:speedup}
The speedup over classical distributed Parareal is bounded as
\[
\frac{\mathcal{C}^{(k)}(p, \mathcal{C}_{F}, \mathcal{C}_{G}, \mathcal{\overline C}_{C,G})}{\mathcal{\widetilde C}^{(\widetilde k)}(p, \mathcal{C}_{F}, \mathcal{C}_{G})} \ \le \ 1 + \frac{\left(p-2\right) \mathcal{\overline C}_{C,G}}{\mathcal{C}_{F} + \mathcal{C}_{G}}.
\]
\end{corollary}
\begin{proof}
If
\[
\mathcal{C}^{(k)}(p, \mathcal{C}_{F}, \mathcal{C}_{G}, \mathcal{\overline C}_{C,G}) < \mathcal{\widetilde C}^{(\widetilde k)}(p, \mathcal{C}_{F}, \mathcal{C}_{G}),
\]
then
\[
\frac{\mathcal{C}^{(k)}(p, \mathcal{C}_{F}, \mathcal{C}_{G}, \mathcal{\overline C}_{C,G})}{\mathcal{\widetilde C}^{(\widetilde k)}(p, \mathcal{C}_{F}, \mathcal{C}_{G})} \ < \ 1 \ \le \ 1 + \frac{\left(p-2\right) \mathcal{\overline C}_{C,G}}{\mathcal{C}_{F} + \mathcal{C}_{G}}, \qquad\qquad p \ge 2.
\]
Now, considering that
\[
\mathcal{C}^{(k)}(p, \mathcal{C}_{F}, \mathcal{C}_{G}, \mathcal{\overline C}_{C,G}) \ge \mathcal{\widetilde C}^{(\widetilde k)}(p, \mathcal{C}_{F}, \mathcal{C}_{G}),
\]
we have
\[
\begin{aligned}
\frac{\mathcal{C}^{(k)}(p, \mathcal{C}_{F}, \mathcal{C}_{G}, \mathcal{\overline C}_{C,G})}{\mathcal{\widetilde C}^{(\widetilde k)}(p, \mathcal{C}_{F}, \mathcal{C}_{G})} & \le \frac{\mathcal{C}^{(k)}(p, \mathcal{C}_{F}, \mathcal{C}_{G}, \mathcal{\overline C}_{C,G}) - p\ \mathcal{C}_{G}}{\mathcal{\widetilde C}^{(\widetilde k)}(p, \mathcal{C}_{F}, \mathcal{C}_{G}) - p\ \mathcal{C}_{G}}\\
& = \frac{k}{\kappa(\widetilde k)} \left(1 + \frac{\left(p-1 - \frac{k+1}{2}\right) \mathcal{\overline C}_{C,G}}{\mathcal{C}_{F} + \mathcal{C}_{G}}\right),
\end{aligned}
\]
where we recall that $p\ \mathcal{C}_{G}$ is the initialization cost. With the maximum gain expected at $\kappa(\widetilde k) = k$, this bound reaches its maximum at $k = 1$, which leads to the result.
\end{proof}

Still, we should point out the fact that, in HPC context, and particularly for time-parallel time-integration (see, e.g., introduction and comments in \cite{FarChan2003, Aub2011, NielHest2016}), one is generally more interested in being able to improve performance by providing more and more processors, in order to keep a given amount of work per processor (strong scaling). In this sense, the bound of 50\% over the traditional parallel efficiency is not actually limiting the desired scaling potential. Let us then highlight, for a more complete performance analysis, that the asymptotic speedup from the distributed Parareal over sequential application of $F$ is given by
\[
\lim_{p \to +\infty} \frac{\mathcal{C}(p, \mathcal{C}_{F})}{\mathcal{C}^{(k)}(p, \mathcal{C}_{F}, \mathcal{C}_{G}, \mathcal{\overline C}_{C,G})} \ = \ \frac{\mathcal{C}_{F}}{\mathcal{C}_{G} + k\ \mathcal{\overline C}_{C,G}} \ \le \ \frac{\mathcal{C}_{F}}{\mathcal{C}_{G} + \mathcal{\overline C}_{C,G}},
\]
while that of the distributed async-Parareal reduces to
\[
\lim_{p \to +\infty} \frac{\mathcal{C}(p, \mathcal{C}_{F})}{\mathcal{\widetilde C}^{(\widetilde k)}(p, \mathcal{C}_{F}, \mathcal{C}_{G})} \ = \ \frac{\mathcal{C}_{F}}{\mathcal{C}_{G}},
\]
which finally implies that
\[
\lim_{p \to +\infty} \frac{\mathcal{C}^{(k)}(p, \mathcal{C}_{F}, \mathcal{C}_{G}, \mathcal{\overline C}_{C,G})}{\mathcal{\widetilde C}^{(\widetilde k)}(p, \mathcal{C}_{F}, \mathcal{C}_{G})} \ = \ 1 + k\ \frac{\mathcal{\overline C}_{C,G}}{\mathcal{C}_{G}}.
\]

As expected, the performance gain from async-Parareal should increase with $\mathcal{\overline C}_{C,G}$, even in a perfectly well balanced configuration, and we saw that $\mathcal{\overline C}_{C,G}$ can be of the same order as $\mathcal{C}_{G}$. It may however seem surprising that the asymptotic speedup does depend on $k$, and not on $\kappa(\widetilde k)$ too. When looking at the classical speedup though (in the proof of Corollary \ref{cor:speedup}), we may expect that, for given $\mathcal{\overline C}_{C,G}$, $F$ and $G$, the ratio $k/\kappa(\widetilde k)$ is bounded by a constant. This is very likely to be true, as the current async-Parareal scheme is also quite close to the classical Parareal one. Further analysis may be needed to fully assess such a relation. Still in this case, while the speedup would decrease with $k$, the asymptotic speedup rather increases. Indeed, in the distributed Parareal algorithm, the amount of communication decreases with $k$ (as partial termination is taken into account), which lessens the impact of $\mathcal{\overline C}_{C,G}$. However, for very high numbers of processors, this diminution is generally negligible, as $p-k$ remains significantly high for a long period of time.

{In the end, while the practical performance gain from applying asynchronous iterations to an existing distributed scheme is no more questionable for spatial parallelism on PDEs, the case of time parallelism is quite less obvious. We therefore present here experimental results only conducted to confirm the possibility of performance gain here as well.}
Especially, such a configuration where each processor depends on only one other processor is probably the worst for squeezing the most out of asynchronous iterations. Furthermore, to give full advantage to the classical distributed Parareal, we experimented using perfect static load balance, on a homogeneous supercomputer with high-speed communication (which minimizes $\mathcal{\overline C}_{C,G}$).
We reported another, more detailed, experimental investigation in \cite{MagEtAl2018}, where we mainly focused, throughout the paper, on implementation aspects within the Message Passing Interface (MPI) framework. This issue about asynchronous iterations programming has been more generally discussed in \cite{MagGBen2017,MagGBen2018}, where an MPI-based programming library has been proposed, including tools to evaluate the loop stopping criterion of Algorithm \ref{algo:pr_aiac} (see \cite{MagGBen2018b} for a full introduction to exact residual error evaluation under asynchronous iterations).

\section{Numerical results}
\label{sec:res}

\subsection{Problem and experimental settings}

The experimental case consisted of simulating an apartment gradually cooled by an air conditioner. We used a simple heat evolution model of the form
\[
\frac{\partial u}{\partial t} = \Delta u,
\]
by setting a constant Dirichlet boundary condition on the air conditioner to simulate a continuous cooling at 73.4 degrees Fahrenheit (23 degrees Celsius), from an initial homogeneous temperature at 86.0 degrees Fahrenheit (30 degrees Celsius). Spatial discrete equations were obtained through $P_1$-Lagrange finite-elements integration upon a mesh composed of 171,478 tetrahedrons and 33,796 nodes (see Figure \ref{fig:apart-mesh}), generated by the TetGen \cite{Si2015} free software. The computer aided design (CAD) was performed using the CATIA V5 proprietary software tool ($\copyright$Dassault Syst\`emes).
\begin{figure}[htbp]
\centering
\includegraphics[width=0.33\textwidth]{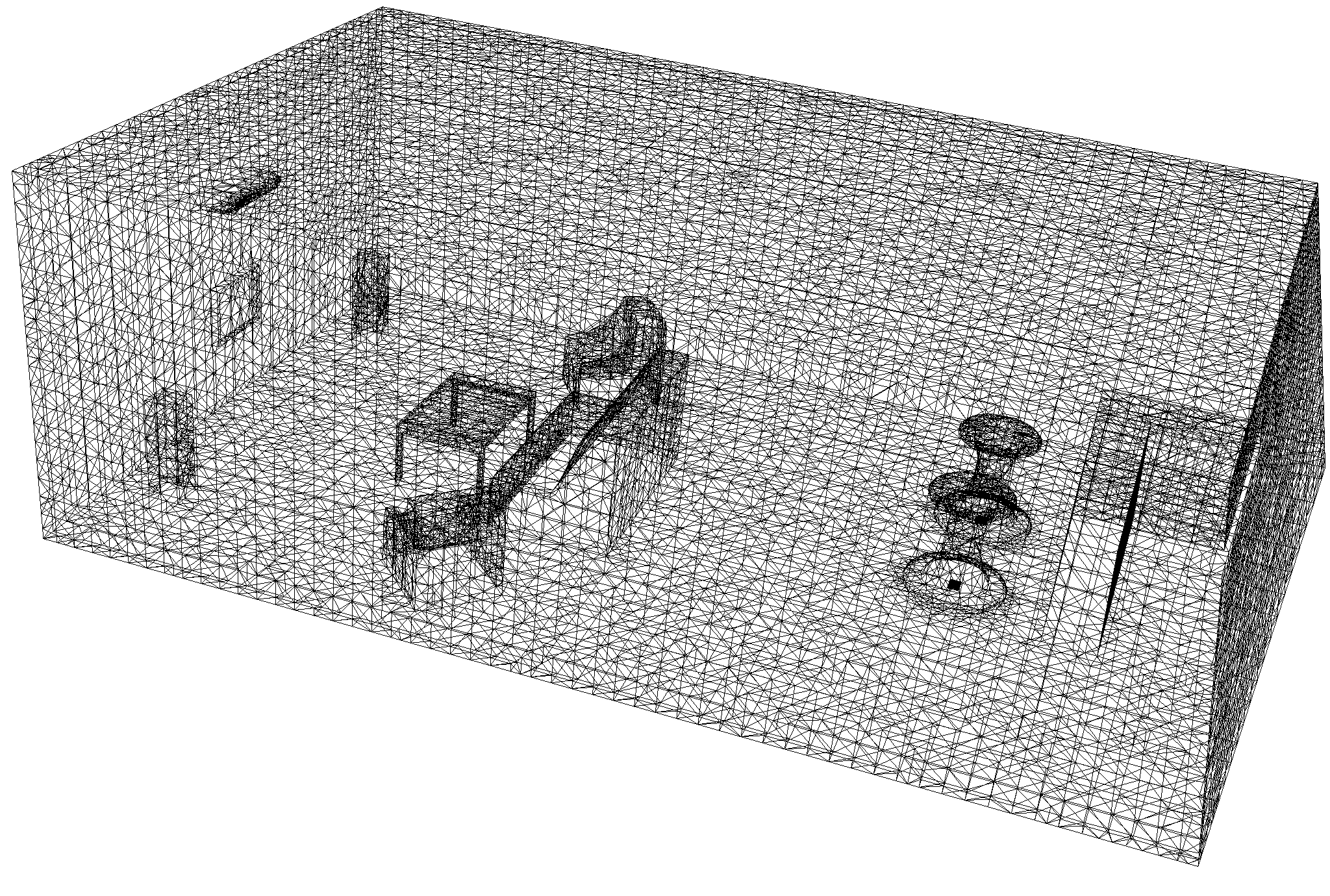}
\quad
\includegraphics[width=0.31\textwidth]{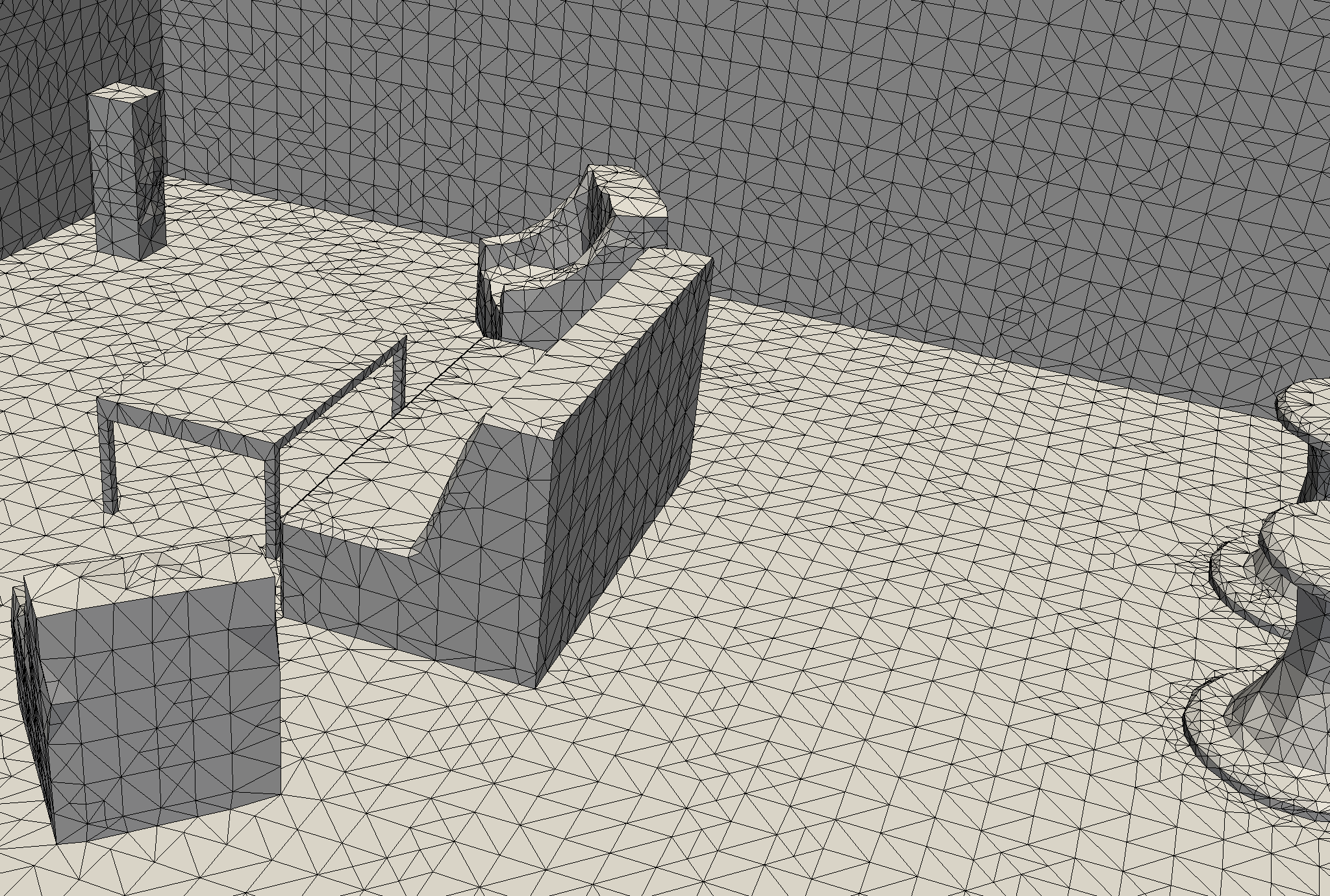}
\caption{Finite-elements mesh of the experimental 3D model.}
\label{fig:apart-mesh}
\end{figure}
Operators $G$ and $F$ consisted of either backward Euler or trapezoidal rule time integration, with respective time steps
\[
\Delta T = 0.2, \qquad \delta t = 0.002,
\]
which implied that
\[
\frac{\mathcal C_{F}}{\mathcal C_{G}} \simeq 100.
\]
Each spatial linear system resolution was sequential and based on LU factorization. Figure \ref{fig:apart-degree} shows the evolution of the temperature over time.
\begin{figure}[htbp]
\centering
\includegraphics[width=0.32\textwidth]{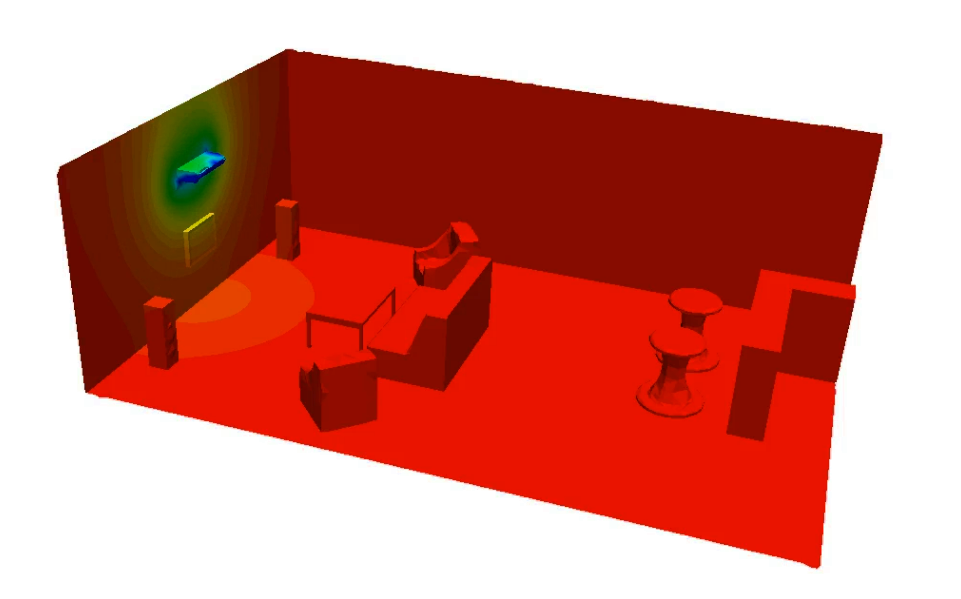}
\includegraphics[width=0.32\textwidth]{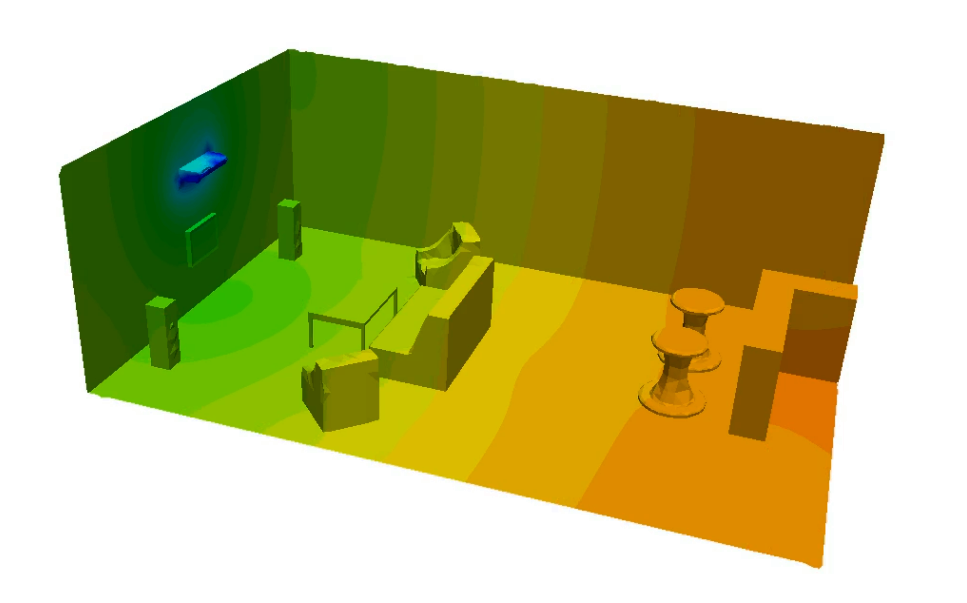}
\includegraphics[width=0.32\textwidth]{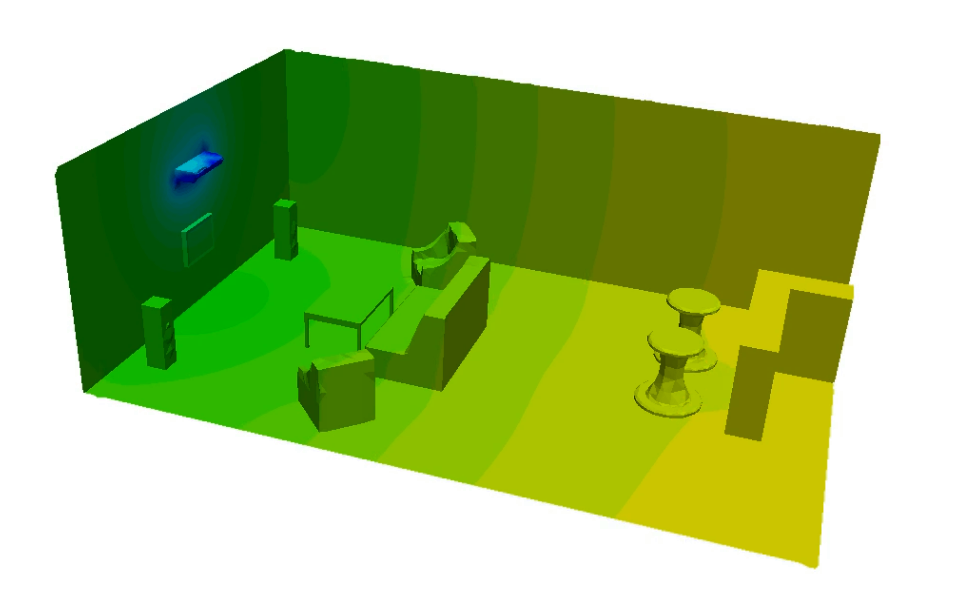}
\caption{Heat distribution at times $t=1$, $t=10$ and $t=20$.}
\label{fig:apart-degree}
\end{figure}

We performed several parallel runs on an SGI Altix ICE supercomputer of 68 nodes interconnected by a QDR Infiniband network (40Gb/s). Each node consisted of two six-cores Intel Xeon X5650 CPUs (12 cores per node) at 2.66GHz, and 21GB RAM allocated for running jobs. The MPI library SGI-MPT was loaded as communication middleware. Execution parameters included the number $p$ of MPI processes, which always equals the number of time intervals, and the loop stopping criterion
{
\begin{equation}
\label{eq:thresh}
\|\lambda^{k}-\lambda^{k-1}\|_{\infty} < \varepsilon, \qquad \varepsilon \in \{10^{-6}, 10^{-5}\}.
\end{equation}
}
We refer to \cite{MagGBen2018} for the implementation of such a criterion in case of asynchronous iterations $k_{1}$ to $k_{p}$.

\subsection{Performance comparison}

{
Table \ref{tab:apart-convergence} and Table \ref{tab:apart-convergence2} report average results, where $\mathcal T$ denotes wall-clock execution times in seconds.
We recall that $T_{p}$ is the simulated end time, and that $\mathcal{\overline C}_{C,G}$ is the average non-overlapped overhead costs induced by communication at each iteration. This does not directly corresponds to communication costs, and therefore is hardly accurately predictable. We rather deduced values $\mathcal{\overline C}_{C,G}$ which make $\mathcal C^{(k)}$ match the measured execution times. We also measured, on average,
\[
\mathcal C_{G} \simeq 0.14 \text{ sec}.
\]
Compared to these $\mathcal{\overline C}_{C,G}$ values, it seems clear that, even in such an ideal HPC computational environment, the actual impact of communication costs cannot be neglected, despite the optimized distributed implementation from Aubanel \cite{Aub2011}.
}
\begin{table}[tbhp]
\caption{Comparison of synchronous and asynchronous Parareal. \newline $F$: trapezoidal rule, $G$: backward Euler, $\varepsilon$: $10^{-6}$.}
\label{tab:apart-convergence}
\centering
{\footnotesize
\begin{tabular}{ccc}
\hline\noalign{\smallskip}
& Parareal & Async-Parareal\\
\noalign{\smallskip}\hline\noalign{\smallskip}
\begin{tabular}{cr}
$p$ & $T_{p}$\\
\noalign{\smallskip}\hline\noalign{\smallskip}
16 & 3.2\\
24 & 4.8\\
32 & 6.4\\
48 & 9.6\\
64 & 12.8\\
90 & 18.0\\
\end{tabular}
&
\begin{tabular}{rrcc}
$\mathcal T$ & $k$ & {$\mathcal{\overline C}_{C,G}$} & $\|\lambda^{k}-\lambda^{*}\|_{\infty}$\\
\noalign{\smallskip}\hline\noalign{\smallskip}
280 & 10 & {1.530} & 1.49E-07\\
420 & 14 & {1.010} & 1.86E-07\\
691 & 20 & {1.011} & 3.75E-08\\
1169 & 32 & {0.727} & 4.95E-11\\
1488 & 44 & {0.486} & 3.75E-08\\
2676 & 64 & {0.486} & 4.75E-11\\
\end{tabular}
&
\begin{tabular}{rrrc}
$\mathcal T$ & $\kappa(\widetilde k)$ & $\mathcal {\widetilde C}^{(\widetilde k)}$ & $\|\lambda^{\widetilde k}-\lambda^{*}\|_{\infty}$\\
\noalign{\smallskip}\hline\noalign{\smallskip}
342 & 24 & 337 & 5.01E-08\\
502 & 34 & 484 & 1.06E-07\\
644 & 44 & 627 & 1.71E-08\\
922 & 66 & 940 & 3.33E-11\\
1255 & 92 & 1310 & 2.79E-10\\
1938 & 149 & 2119 & 4.78E-11\\
\end{tabular}\\
\noalign{\smallskip}\hline
\end{tabular}
}
\end{table}
\begin{table}[tbhp]
\caption{Comparison of synchronous and asynchronous Parareal. \newline $F$: backward Euler, $G$: backward Euler, $\varepsilon$: $10^{-5}$.}
\label{tab:apart-convergence2}
\centering
{
{\footnotesize
\begin{tabular}{ccc}
\hline\noalign{\smallskip}
& Parareal & Async-Parareal\\
\noalign{\smallskip}\hline\noalign{\smallskip}
\begin{tabular}{cr}
$p$ & $T_{p}$\\
\noalign{\smallskip}\hline\noalign{\smallskip}
16 & 3.2\\
24 & 4.8\\
32 & 6.4\\
48 & 9.6\\
\end{tabular}
&
\begin{tabular}{rrcc}
$\mathcal T$ & $k$ & $\mathcal{\overline C}_{C,G}$ & $\|\lambda^{k}-\lambda^{*}\|_{\infty}$\\
\noalign{\smallskip}\hline\noalign{\smallskip}
116 & 6 & 0.233 & 4.43E-07\\
131 & 6 & 0.307 & 4.99E-07\\
126 & 6 & 0.211 & 3.41E-07\\
139 & 6 & 0.196 & 1.32E-07\\
\end{tabular}
&
\begin{tabular}{rrrc}
$\mathcal T$ & $\kappa(\widetilde k)$ & $\mathcal {\widetilde C}^{(\widetilde k)}$ & $\|\lambda^{\widetilde k}-\lambda^{*}\|_{\infty}$\\
\noalign{\smallskip}\hline\noalign{\smallskip}
348 & 18 & 293 & 4.68E-08\\
379 & 19 & 297 & 6.44E-08\\
382 & 19 & 276 & 9.35E-08\\
360 & 15 & 211 & 1.48E-07\\
\end{tabular}\\
\noalign{\smallskip}\hline
\end{tabular}
}
}
\end{table}

{
As expected, Table \ref{tab:apart-convergence} shows that the overhead costs at one iteration tend to diminish as $k$ grows, which makes the average $\mathcal{\overline C}_{C,G}$ diminish as well, while the sum of these costs can however keep increasing, which explains the increasing better performance of asynchronous execution. On the other hand, estimated values $\mathcal {\widetilde C}^{(\widetilde k)}$ are quite close to measured times (mostly in Table \ref{tab:apart-convergence}). This is explained by the fact that asynchronous execution somehow converts overhead communication costs into additional iterations which are accounted by $\kappa(\widetilde k)$. Dealing with the prediction of $\kappa(\widetilde k)$ could be a less harassing task, compared to $\mathcal{\overline C}_{C,G}$.
One can notice for instance that
\[
2 < \frac{\kappa(\widetilde k)}{k} < 3, \qquad\qquad 2 < \frac{\kappa(\widetilde k)}{k} < 4,
\]
in Table \ref{tab:apart-convergence} and Table \ref{tab:apart-convergence2} respectively, which confirms, for these test cases, our expectation of a generally bounded ratio.
}

Finally, with a stopping criterion \eqref{eq:thresh}, it happened that async-Parareal generally tended to provide a better precision, according to the sequential solution from $F$. While this was probably due to the particular termination delay of asynchronous iterations, it did not prevent us, in case of Table \ref{tab:apart-convergence}, from obtaining a faster solver, for $p \ge 32$. For instance, around 12 minutes are saved over 44, using 90 processors.
{
On the contrary, Table \ref{tab:apart-convergence2} shows cases where, with sufficiently low values $\mathcal{\overline C}_{C,G}$ (regarding numbers of iterations), synchronous iterations keep performing better than asynchronous ones.
}

\section{Conclusion}
\label{sec:conclu}

Efficiently applying asynchronous iterations in time domain decomposition is rather challenging, and this primary work certainly leaves an amount of further desirable {benchmarks and} improvements. While we tackled here the basic Parareal scheme, it is somehow important to note that introducing asynchronous iterations is a design approach quite different from, and not concurrent to, classical scheme improvement, as they could actually be applied to more advanced Parareal schemes as well, as long as communication is a possible factor of performance limit in real situations. What is particularly interesting is that convergence conditions (in maximum error norm) barely change, and are even asymptotically the same, while also finite-time termination remains guaranteed. Still, despite the potential theoretical performance gain and some experimental results confirming it in practice, it is clear that such a pipeline parallel configuration drastically limits the benefits from asynchronous iterations,
{compared to spatial parallelism where each processor generally depends on several others, such as in undirected grid configurations.}
Hopefully, as pointed out by this first theoretical analysis, much more flexibility is now provided to deal with this issue, which is currently under consideration.

\bibliography{ref}
\bibliographystyle{abbrv}

\end{document}